\documentclass[11pt,openright]{article}

\usepackage{amsthm}
\usepackage{amsfonts}
\usepackage{amssymb}
\usepackage{amsthm}
\usepackage{amsmath}    
\usepackage{hyperref}
\usepackage{bbm}
\usepackage{graphicx}
\usepackage{latexsym}
\usepackage{mathrsfs}
\usepackage{bbm}
\usepackage{slashed}
\usepackage{mathabx}
\usepackage{xcolor}
\usepackage{moodle}
\usepackage[normalem]{ulem} 
\usepackage{geometry}
\DeclareMathAlphabet\EuFrak{U}{euf}{m}{n}	
\SetMathAlphabet\EuFrak{bold}{U}{euf}{b}{n}	

\newcommand{\cA}{\mathcal{A}}

\newcommand{\cC}{\mathcal{C}}
\newcommand{\cD}{\mathcal{D}}

\newcommand{\cH}{\mathcal{H}}

\newcommand{\cK}{\mathcal{K}}

\newcommand{\cM}{\mathcal{M}}

\newcommand{\cO}{\mathcal{O}}

\newcommand{\cR}{\mathcal{R}}

\newcommand{\cV}{\mathcal{V}}
\newcommand{\cW}{\mathcal{W}}

\newcommand{\cbH}{\boldsymbol{\mathcal{H}}}

\newcommand{\cbK}{\boldsymbol{\mathcal{K}}}

\newcommand{\rC}{\mathrm{C}}

\newcommand{\rO}{\mathrm{O}}

\newcommand{\bR}{\mathbb{R}}
\newcommand{\bS}{\mathbb{S}}

\newcommand{\si}{\sigma}

\newcommand{\eps}{\varepsilon}

\newcommand{\ch}{\mathrm{ch}}

\newcommand{\R}{\mathbb{R}} 

  \theoremstyle{plain}
  \newtheorem{definition}{Definition}[section]
  \newtheorem{theorem}[definition]{Theorem}
  \newtheorem{proposition}[definition]{Proposition}
  \newtheorem{corollary}[definition]{Corollary}
  \newtheorem{lemma}[definition]{Lemma}
  
  \theoremstyle{definition}

\title{\Huge{Relative entropy and curved spacetimes} }

\author{{\sc Fabio Ciolli}, {\sc Roberto  Longo},  {\sc Alessio  Ranallo}, {\sc Giuseppe  Ruzzi}\\
Dipartimento di Matematica,
Universit\`a di Roma Tor Vergata\\
Via della Ricerca Scientifica, 1, I-00133 Roma, Italy\footnote{Supported by the ERC Advanced Grant 669240 QUEST ``Quantum Algebraic Structures and Models'', MIUR FARE R16X5RB55W  QUEST-NET and GNAMPA-INdAM. \eject
} }

\date{}
\begin{document}

\maketitle

\begin{abstract}
Given any half-sided modular inclusion of standard subspaces, we show that the  entropy function associated with the decreasing one-parameter family of translated standard subspaces is convex for any given (not necessarily smooth) vector in the underlying Hilbert space. 
In second quantisation, this infers the convexity of the vacuum relative entropy with respect to the translation parameter of the modular tunnel of von Neumann algebras. 
This result allows us to study the QNEC inequality for coherent states in a free Quantum Field Theory on a stationary curved spacetime, given a KMS state. 
To this end, we define wedge regions and appropriate (deformed) subregions. Examples are given by the Schwarzschild spacetime and null  translated subregions with respect to the time translation Killing flow. More generally, we define wedge and strip regions on a globally hyperbolic spacetime, so to have non trivial modular inclusions of von Neumann algebras, and make our analysis in this context. 
\end{abstract}

\newpage


\section{Introduction}
Entropy/energy inequalities lie at the basis of Physics and are currently a subject of great interest in Quantum Field Theory, with motivations coming from diverse sources as black hole thermodynamics, quantum information, conformal field theory. The Quantum Null Energy Condition, QNEC, states that the second derivative of the relative entropy is non-negative
\[
\frac{d^2}{d\lambda^2} S(\lambda) \geq 0
\]  
when one shrinks a wedge region in the null direction ($\lambda$ is a deformation parameter), see \cite{BFLW15}. Here $S$ is the relative entropy of the associated observable von Neumann algebra between the vacuum state and any other normal state, see \cite{Ar}. This a thought-provoking inequality, untrue for general classical or quantum systems, whose ultimate nature is not yet fully understood. Certainly, quantum aspects and relativistic invariance play a role and Operator Algebras provide the right mathematical framework to describe it. 

Now, the classical Null Energy Condition is relevant in General Relativity, so one is naturally led to check the QNEC for Quantum Field Theory on a Curved Spacetime. The purpose of this paper is formulate and provide an analysis of the QNEC in this general context. Indeed we shall give a very general result that can be applied to a variety of situations. 

In the second quantisation Minkowski context, a formula for the vacuum relative entropy of a coherent state has been defined in \cite{L19, CLR} and is expressed in first quantisation. The general formula, recalled here below \eqref{Sintro}, depends on the symplectic structure of the one-particle space, the metric structure plays a role only about the choice of the unitary evolution associated with the Killing flow. This strongly suggests  the analysis in \cite{CLR}  be naturally done within a curved spacetime framework.  

Our basic result concerns inclusions of standard subspaces of a complex Hilbert space $\cH$. Recall that a standard subspace $H$ of $\cH$ is a closed, real linear subspace such that
\[
H \cap i H = \{0\}\, ,\quad \overline{H + iH} = \cH\, .
\]
An inclusion of standard subspaces $K \subset H$ is said to be half-sided modular if
\[
\Delta_H^{-is}K \subset K\, ,\quad s\geq 0\, ,
\]
with $\Delta_H$ the modular operator of $H$, see \cite{L}. Then one has a monotone family of translated standard subspaces $H_s$, $s\in \mathbb R$, such that $\Delta_H^{-it}K = H_{e^{2\pi t}}$. 

Now, given a vector $\phi\in \cH$, the entropy of $\phi$ with respect to the standard subspace $H$ is defined by
\begin{equation}\label{Sintro}
S^H_\phi = - \Im(\phi , P_H i \log\Delta_H \phi)
\end{equation}
(as quadratic form), with $P_H$ the cutting projection $H+H'\to H$ \cite{CLR}, see Sect. \ref{Aa}. 
Here $H' \equiv i H^{\perp}$ is the symplectic complement of $H$, where $H^{\perp}$ is the orthogonal of $H$ w.r.t. to the real part of the scalar product. 
$S^H_\phi$ is finite for a dense linear subspace of $\cH$ (for the moment we assume $H$ to be factorial, i.e.\ $H\cap H' = \{0\}$).

If $K\subset H$ is half-sided modular, we may consider the function 
\[
S^{\cbH}_\phi: \lambda\mapsto   S^{H_\lambda}_\phi \, , 
\]
with $\cbH$ the triple $(\cH, H,K)$. 
Due to general properties of the entropy, $S^{\cbH}_\phi(\lambda) \geq 0$ or $S^{\cbH}_\phi(\lambda) = +\infty$ and 
$S^{\cbH}_\phi(\lambda) $ is non decreasing. 

Our main abstract result is that $S^{\cbH}_\phi(\lambda)$ is a convex function for every $\phi\in \cH$. Indeed, if for $\lambda_0\in\mathbb R$ we have $S^{\cbH}_\phi(\lambda_0) < \infty$, then 
$S^{\cbH}_\phi(\lambda)$ is finite $C^1$ on $[\lambda_0, \infty)$ and, on this interval,
 $\frac{d}{d\lambda}S^{\cbH}_\phi$ is absolutely continuous with 
$\frac{d^2}{d\lambda^2}S^{\cbH}_\phi\geq 0$  almost everywhere. 

The key point about the entropy $S^H_\phi$ of a vector $\phi$ with respect to the standard subspace $H$ is that it gives Araki's relative entropy
\[
S^H_\phi = S^{\cR_\varphi(H)}(\varphi_\phi |\!| \varphi)
\]
between the vacuum state $\varphi$ and the coherent state $\varphi_\phi$ associated with $\phi\in\cH$, on the von Neumann algebra $\cR_\varphi(H)$, on the Bose Fock space, generated by the Weyl unitaries associated with vectors in $H$.

It is now clear how the result immediately gives the QNEC inequality for coherent vectors in a free Quantum Field Theory in a large class of contexts: it suffices to have a globally hyperbolic spacetime, a subregion mapped into itself by a timelike Killing flow in the positive direction
and a KMS state for the observable operator algebra.
We shall illustrate this geometric operator algebraic structure in the last section, with various examples. We introduce the notions of wedge and strip regions of a globally hyperbolic spacetime and make an analysis in this general framework. 

\section{Entropy and modular inclusions}
In this section we make an abstract analysis that will later be applied in the context of Quantum Field Theory on a curved spacetime. 

\subsection{Entropy and standard subspaces}
\label{Aa}
We recall a few facts about standard subspaces and the notion of entropy of a vector with respect to a standard 
subspace. References for this section are \cite{L,CLR}. 

Given a complex Hilbert space $\cH$, a closed real linear subspace $H$ is said to be \emph{standard} if $H\cap iH=\{0\}$ and $\overline{H+iH}=\cH$. 
To  a standard subspace $H$ one associates an involutive, closed, anti-linear operator, the \emph{Tomita operator} $S_H$ 
\[
S_H(\phi+i\psi) = \phi-i\psi\, ,\ \phi,\psi\in H\, ,
\] whose polar decomposition $S_H = J_H \Delta_H$ gives 
an antiunitary  operator $J_H$, \emph{the modular conjugation}, and a self-adjoint, positive, nonsingular 
operator $\Delta_H$, \emph{the modular operator},  satisfying the relations
\[
  J_H=J^*_H \ , \ \ J_H\Delta_H J_H= \Delta^{-1}_H\ .
\]
The \emph{modular unitary group} is the one-parameter unitary group $\Delta^{is}_H$ that verifies 
\[
\Delta^{is}_H H= H  \ , \ \ J_H H = H'   \ , 
\] 
where $H'$ denotes the \emph{symplectic complement} of $H$ given by 
\[
H' = \{\psi\in\cH: \Im(\psi,\phi)=0,\, \forall  \phi\in H\}\ .
\]
We have: 
\begin{itemize}
\item If $U$ is a unitary operator of $\cH$ then $K=UH$ is  a standard space and 
\[
\Delta_K=U\Delta_HU^* \ \ , \ \ J_K= UJ_HU^*\, .
\]
\item If $H,K$ are standard subspaces of the Hilbert spaces $\cH$ and $\cK$, then $H\oplus K$ is a standard space 
of $\cH\oplus \cK$ and  
\[
J_{H\oplus K}= J_H\oplus J_K \ \ , \ \  \Delta_{H\oplus K}= \Delta_H \oplus \Delta_K\, .
\]
\item $H\cap H'$ is equal to ker$(1-\Delta_H)$.
\end{itemize} 
 We say that $H$ is \emph{factorial} if
 $H\cap H' = \{0\}$. 
 Given a factorial standard  subspace $H$ of $\cH$, the \emph{cutting projection} associated with $H$  is defined as
\[
P_H(\phi+\phi')= \phi \ , \qquad \phi\in H, \ \phi'\in H' \, .
\]
It turns out that $P_H$ is a densely defined, closed, real linear operator satisfying  
\[
P^2_H = P_H \  \ , \ \ -iP_H i=P_{iH} \ \ , \ \  P_H\Delta^{it}_H=
\Delta^{it}_H P_H \ .
\]
In general, if $H$ is not factorial, we have a direct sum decomposition
\begin{equation}\label{Sf}
H = H_a \oplus H_f\, ,
\end{equation}
where $H_a = H\cap H'$ and $H_f$ is the real orthogonal of $H_a$. So $ \cH = \cH_a \oplus \cH_f$ where 
$H_a\subset \cH_a$ is an Abelian standard subspace and $H_f\subset\cH_f$ is a factorial standard subspace of $\cH_f$ (see also \cite{BCD}).

If $H$ is Abelian, following \cite{BCD}, we recall that the entropy of $\phi$ with respect to $H$ is
\begin{equation}\label{Eab}
S^{H}_\phi = 2(\phi, (1 -E)\phi) \, ,
\end{equation}
with $E$ the real orthogonal projection of $\cH = H + i H$ onto $H$. 

In general, if $H\subset \cH$ is any standard subspace, we consider the factorial decomposition $H_a \oplus H_f\subset \cH_a\oplus \cH_f$ \eqref{Sf}
and define the entropy of $\phi$ with respect to $H$ as
\begin{equation}\label{Entrgen}
S^{H}_\phi \equiv S^{H_a}_{\phi_a} + S^{H_f}_{\phi_f}\, ,
\end{equation}
where $\phi = \phi_a \oplus\phi_f$ is the decomposition of $\phi$,
$S^{H_a}_{\phi_a}$ is the relative entropy of $\phi_a\in \cH_a$ with respect $H_a$  and similarly for $S^{H_f}_{\phi_f}$.

If $H\subset \cH$ is any closed, real linear subspace, we set $S^H_\phi = S^{H_s}_{\phi_s}$, where $H_s$ is the standard component of $H$ in $\cH_s =(H\cap iH)^{\perp_\mathbb R}\cap \overline{H+iH}$. 
We thus may, and will, assume the considered subspaces to be standard and factorial. 

Some of the main properties of the entropy of a vector are:
\begin{itemize}
\item $S^H_\phi\geq 0$ or $S^H_\phi = +\infty$ (\emph{positivity}); 
\item If $K \subset H$, then $S^K_\phi\leq S^H_\phi$ (\emph{monotonicity});
\item If $\phi_n\to \phi$, then $S^H_\phi \leq \liminf_n S^H_{\phi_n}$ (\emph{lower semicontinuity});
\item If $H_n\subset H$ is an increasing sequence with $\overline{\bigcup_n H_n} = H$, then $S^{H_n}_\phi \to S^H_{\phi}$
(\emph{monotone continuity}).
\end{itemize}
It turns out that \cite{CLR}
\begin{equation}\label{Sfin}
S^H_\phi <+\infty \iff -\int^1_0 \log\lambda\, d(\phi,E(\lambda)\phi)<+\infty \iff  
\phi\in D(\sqrt{|\log\Delta_H|E_-})\, .
\end{equation}
Here $E(\lambda)$ are the spectral projections of $\Delta_H$ and $E_-$ is the spectral projection onto the negative part of the spectrum of $\log\Delta_H$. 

Now, let $H$ be a standard subspace of a  Hilbert space. Denote by $\cR_\varphi(H)$ the von Neumann algebra 
generated by the Weyl unitaries associated with vectors in $H$ on 
the second quantisation Bose-Fock space over $\cH$, and $\varphi$ is the vacuum state  (see Section \ref{QF1}).  
\begin{proposition}
\label{equiv-entropy}
For any standard subspace $H$ of a Hilbert space $\cH$ and $\phi\in\cH$,  $S^H_\phi $  equals the Araki's relative entropy 
\[
S^H_\phi = S^{\cR_\varphi(H)}(\varphi_\phi |\!| \varphi)
\]
between the restriction of the vacuum state $\varphi$ and of the coherent state $\varphi_\phi$ associated with $\phi\in\cH$ to the von Neumann algebra $\cR_\varphi(H)$.
\end{proposition}
\begin{proof}
This proposition is \cite[Thm. 4.5]{CLR} in the factorial case. In the Abelian case the proposition is proved in \cite{BCD}. 
It then holds in general by considering the factorial decomposition \eqref{Sf}. Indeed, 
\[
\varphi_\phi = \varphi_{\phi_a\oplus \phi_f} = 
\varphi_{\phi_a}
\otimes \varphi_{\phi_f} 
\]
entails our statements by the additivity of the relative entropy under tensor product. 
\end{proof}

\subsection{Half-sided modular inclusions}
\label{Ab}

We now discuss the properties of the relative entropy of a fixed vector with respect to the one-parameter family of standard subspaces associated with a given half-sided modular inclusion. In particular, we analyse the convexity property. 

Let $\cbH=(\cH, H_0,H_1)$ be an \emph{half-sided modular inclusion} i.e.\, $\cH$ is a complex Hilbert space,
$H_1\subset H_0$ is an inclusion of standard subspaces of $\cH$ and we have
\[
\Delta^{-is}_{H_0}H_1 \subset H_1 \ , \qquad s \geq 0\, .
\] 
By \cite{AZ05, L},  the operator 
$\frac1{2\pi}(\log\Delta_{H_1}-\log\Delta_{H_0})$ 
is  essentially self-adjoint  on $D(\log\Delta_{H_1})\cap D(\log\Delta_{H_0})$ with  positive  closure that we denote by 
$X$. The one-parameter unitary group generated by $X$, 
\emph{the translation unitary group} $U(t)=\exp(itX)$, satisfies the relations 

$a)$ $\Delta^{-is}_{H_0} U(t)\Delta^{is}_{H_0}=U(e^{2\pi s}t)$ and  
$J_{H_0} U(t) J_{H_0}=U(-t)$,  $t\in\bR$; 

$b)$ $U(t) H_0\subset H_0$ for  $t\geq 0$; 

$c)$ $U(1) H_0= H_1$. 

\noindent
In particular, last relation  gives 
\begin{equation}
\label{fv}
\Delta^{is}_{H_1}= U(1)\Delta^{is}_{H_0}U(-1) \ \ , \ \ J_{H_1}= U(1) J_{H_0}U(-1) \ . 
\end{equation} 
 Let $G_0$ be the ``$ax + b$'' group, namely the group of diffeomorphisms of $\mathbb R$ generated by translations and dilations, and let $G$ be the improper ``$ax + b$'' group, namely the group generated by $G_0$ and the reflection $x \mapsto -x$ on $\mathbb R$.  
Relation $a)$ says that $\cbH$ gives a positive energy, anti-unitary representation of $G$. The converse is true \cite{BGL02}, namely every positive energy unitary representation of $G$ arises in this way. So,
there exists a one-to-one correspondence
\[
\text{Half-sided modular inclusion} \longleftrightarrow  \text{Positive energy anti-unitary representation of $G$}
\]
This correspondence concerns  representations, not merely unitary equivalence classes. 
Representations of $G$ with strictly positive energy correspond to factorial half-sided modular inclusions, trivial representations of $G$ to trivial inclusions, i.e.\ $H_0 = H_1$.  Indeed, the following holds:%
\begin{proposition}
\label{factorial}
Let $\cbH=(\cH, H_0,H_1)$ be a  half-sided modular inclusion and $H_0 = H_{k,a} \oplus H_{k,f}\subset \cH_{k,a} \oplus \cH_{k,f}$ the Abelian/factorial decomposition of $H_k$, $k =0,1$. Then $\cH_{0,a} = \cH_{1,a}$,  $\cH_{0,f} = \cH_{1,f}$ and  
\[
H_a \equiv H_{0,a} = H_{1,a}\, .
\]
$H_a$ is the fixed-point subspace for the associated unitary representation of $G$. 
\end{proposition}
\begin{proof}
Let $H_{a}$ be the $\Delta^{is}_{H_0}$ fixed-point subspace, $s\in \mathbb R$, thus $H_a \equiv H_{0,a}$. 
Then $H_a$ is pointwise left fixed by $U$ \cite[Prop. B.3]{GL96}, with $U$ the associated, positive energy, representation of $G$. By eq. \eqref{fv}, $H_{a}$ is then also left pointwise fixed by $\Delta^{it}_{H_1}$, so $H_{0,a} \subset H_{1,a}$. By repeating the argument exchanging $H_{0,a} $ and $ H_{1,a}$, we have $H_{1,a} \subset H_{0,a}$, hence the thesis because $U$ has no non-zero fixed vector on $H_{0,f}$ as this is a factorial standard space.
\end{proof}
\begin{lemma}
Let $\cbH=(\cH, H_0,H_1)$ and $\cbK=(\cK, K_0,K_1)$ be half-sided modular inclusions and $T:\cH \to\cK$ be a bounded, complex linear operator. The following are equivalent

$(i)$ $T\Delta^{is}_{H_0}= \Delta^{is}_{K_0}T$, \ $T\Delta^{is}_{H_1}= \Delta^{is}_{K_1}T$\, ,\ $TJ_{H_0} = J_{K_0}T$,\,
\,  $t\in\mathbb R$. 

$(ii)$ $T$ intertwines the anti-unitary representations of $G$ associated with $\cbH$ and $\cbK$. 
\end{lemma}
\begin{proof}
This follows immediately because the anti-unitary representation of $G$ on $\cH$ is generated by the modular unitary groups  and modular conjugation of $H_0$, $H_1$ and similarly for $\cbK$.  
\end{proof}
An \emph{intertwiner} $T$ between two half-sided modular inclusions $\cbH$ and $\cbK$
is a bounded, complex linear operator $T:\cH\to \cK$ as in the previous lemma. 
We denote the set of the intertwiners 
from $\cbH$ to $\cbK$ by $(\cbH,\cbK)$ and say that: $\cbH$ and $\cbK$ are  
\emph{unitary equivalent} if there is a unitary operator $U\in(\cbH,\cbK)$; 
$\cbH$ is \emph{irreducible} if the \emph{commutant} $(\cbH,\cbH)$ equals 
$\mathbb{C}$. Given a family of half sided modular inclusions $\cbH_\ell$, their 
\emph{direct sum}  is the half-sided modular inclusion defined by
\[
\bigoplus_\ell \cbH_\ell = \big(\oplus_\ell \cH_\ell, \oplus_\ell {H_0}_\ell, \oplus_\ell {H_1}_\ell\big)\, .
\]
As is well known, $G_0$ has only one irreducible unitary representation (up to unitary equivalence) with strictly positive energy, the \emph{Schr\"odinger representation}; this follows by von Neumann uniqueness theorem on the canonical commutation relations. Every irreducible unitary representation of $G_0$ with strictly positive energy is thus a multiple of the Schr\"odinger representation. Now, the Schr\"odinger representation of $G_0$ extends to an irreducible representation of $G$ on the same Hilbert space (see \cite{L}), we call it with the same name. The following proposition is essentially proved in \cite{L}. 
\begin{proposition}\label{CShU}
Let $U$ be an anti-unitary representation of  $G$  with strictly positive energy. Then $U$ is a multiple of the Schr\"odinger representation. Namely
\begin{equation}
U = \bigoplus_{\ell}  U_\ell
\end{equation}
where each $U_\ell$ is unitary equivalent to  the Schr\"odinger representation. 
\end{proposition}
\begin{proof}
Let $U_0$ be the restriction of $U$ to $G_0$. Thus $U_0$ is unitarily equivalent to $V_0\otimes 1$ on $\cH\otimes\cK$, where
with $V_0$ the  Schr\"odinger representation of $G_0$ on $\cH$ and $\cK$ a Hilbert space. With $V$ the the  Schr\"odinger representation of $G$ on $\cH$ extending $V_0$,
let $J_0$ be the anti-unitary involution on $\cH$ corresponding to the reflection $x\to -x$ by $V$, and $J_1$ an arbitrary anti-unitary involution on $\cH$. 
By identifying the Hilbert spaces of $U$ and $V_0\otimes 1$, we have to prove that if $J$ is a anti-unitary 
involution on $\cH\otimes\cK$ commuting with $V_0\otimes 1$, then there exists a unitary $T$ on $\cH\otimes\cK$  commuting with $V_0$ such that $J_0\otimes J_1 = TJT^*$. Now, $J(J_0\otimes J_1)$ commutes with $V_0$, so 
$J = J_0\otimes ZJ_1$ 
for some unitary $Z$ on $\cK$ with $J_1 Z J_1 = Z^*$. Since any two  anti-unitary 
involutions on $\cK$ are unitarily equivalent, the result follows 
\end{proof}
So, there is a unique, up to unitary equivalence, irreducible, half-sided modular inclusion whose associated positive energy unitary representation of  $G$  is the Schr\"odinger one; we call it the \emph{Schr\"odinger half-sided modular inclusion}. 
\begin{corollary}\label{CShH}
Let $\cbH$ be a factorial half-sided modular inclusion. Then $\cbH$ is a multiple of the Schr\"odinger representation. Namely
\begin{equation}\label{Ug}
\cbH = \bigoplus_{\ell}  \cbH_\ell
\end{equation}
where each $\cbH_\ell$ is unitary equivalent to  the Schr\"odinger representation. 
\end{corollary}
\begin{proof}
Immediate by the above discussion.
\end{proof}
Clearly, if $\cbH$ is not factorial, then $\cbH$ is the direct sum of a multiple of the Schr\"odinger representation and a multiple of the trivial inclusion. 

\subsection{Entropy convexity in the modular family parameter}
\label{Ac}

\begin{lemma}\label{Lop} With $\cbH, \cbK, \cbH_\ell$ half-sided modular inclusions of standard subspaces,
we have

$(i)$ $S^{K}_{U\phi} =  S^{H}_{\phi}$ for any unitary $U:\cH\to \cK$ such that $UH = K$  and $\phi\in \cH$.

$(ii)$ If $\cH = \bigoplus_{\ell} \cH_\ell$ and $H = \bigoplus_{\ell} H_\ell$, then
\[
S^{H}_{\phi} = \sum_{\ell} S^{H_\ell}_{\phi_\ell}
\]
for any vector $\phi =\bigoplus_{\ell}\phi_\ell \in\cH$.
\end{lemma}
\begin{proof}
$(i)$ is immediate. $(ii)$ easily follows by monotone continuity. 
\end{proof}
Now, given a half-sided modular inclusion $\cbH=(\cH,H_0,H_1)$, by means of the associated translation one parameter unitary group,
we can define the  \emph{translated standard subspaces}: 
\begin{equation}
H_\lambda = U(\lambda) H_0 \ ,  \qquad   \lambda \in\mathbb R .
\end{equation}
Note that  
\[
 H_{\lambda_2}\subset H_{\lambda_1} \  , \qquad \lambda_2 >\lambda_1 \in\mathbb R  \ ,
\]
Clearly,
\[
 \Delta^{it}_\lambda = U(\lambda)\Delta^{it}_0 U(-\lambda) \ \ , \ \ 
 J_\lambda = U(\lambda)J_0 U(-\lambda) \ ,
\]
where $\Delta_\lambda$, $J_\lambda$ are the modular operator and conjugation of $H_\lambda$. 
We have 
\[
 U_{\cH}(\lambda)P_{H_0}=P_{H_\lambda} U_{\cH}(\lambda) \, ,
\]
Our aim is to study the convexity of the  function
given by the  translated subspaces $H_\lambda$ 
\[
S^{\cbH}_\psi(\lambda)=  - \Im(\psi,P_{H_\lambda} i\log\Delta_{H_\lambda} \psi)  \, ,
\]
for any given vector $\psi \in \cH$, in the quadratic form sense (see [15] for the definition in case $\psi$ is not in the domain of $\log \Delta_H$). 

\begin{lemma}\label{DS}
We have

$(i)$ $S^{\cbK}_{U\phi}(\lambda)=  S^{\cbH}_{\phi}(\lambda)$ for any unitary $U\in(\cbH,\cbK)$  and $\phi\in \cH$.

$(ii)$ If $\cH = \bigoplus_{\ell} \cH_\ell$, then
\[
S^{\cbH}_{\phi}(\lambda)= \sum_{\ell} S^{\cbH_\ell}_{\phi_\ell}(\lambda)
\]
for any vector $\phi =\bigoplus_{\ell}\phi_\ell \in\cH$.
\end{lemma}
\begin{proof}
Immediate by Lemma \ref{Lop}.
\end{proof}
\begin{theorem}\label{Sl}
Let $\cbH$ be a half-sided modular inclusion and $\phi\in\cal H$ a vector. The function
\[
S^{\cbH}_\phi: \lambda \in \mathbb R\to 
S^{H_\lambda}_\phi \in [0, \infty]
\] is convex in $[0,\infty)$. 

Suppose further that $S^{H_{\lambda_0}}_\phi < \infty$ for some $\lambda_0\in\mathbb R$. Then
\begin{itemize}
\item[(i)] $S^{\cbH}_\phi(\lambda)$ is finite and $C^1$ on $[\lambda_0, \infty)$;
\item[(ii)] $\frac{d}{d\lambda}S^{\cbH}_\phi(\lambda)$ is absolutely continuous in $[\lambda_0, \infty)$ with almost everywhere non-negative derivative
$\frac{d^2}{d\lambda^2}S^{\cbH}_\phi(\lambda)\geq 0$\ . \end{itemize}
\end{theorem}
\begin{proof} 
By considering the factorial decomposition of $\cbH$, since the Abelian part is fixed by the translation unitaries by Prop. \ref{factorial}, we may assume that $\cbH$ is factorial.

The convexity statement follows by Lemmas \ref{Lop}, \ref{DS} and \ref{ShC} because the (finite or infinite) sum of convex functions is convex. 

The proof of the remaining statements relies on Theorem \ref{STh}. If $S^{\cbH}_\phi(\lambda_0)$ is finite, then $S^{\cbH}_\phi(\lambda)$ is finite too, $\lambda\geq \lambda_0$, by the monotonicity of the entropy. By replacing $H$ with $H_{\lambda_0}$, we may assume that $\lambda_0 = 0$. 
Now, up to unitary equivalence, $\cbH = \bigoplus_\ell \cbH_\ell$, with $\cbH_\ell$ the Schr\"ordinger representation, see Section \ref{App}. So we may assume that $\cbH_\ell$ is actually the Schr\"ordinger representation.

With $\phi =\oplus\phi_\ell$ the decomposition of $\phi$, we have ${\phi'_\ell}^2\in L^1(\mathbb R_+, dx)$ and
\[
S^{\cbH_\ell}_{\phi_\ell}(\lambda) = 
 \pi\!\int_{\lambda}^{+\infty} (x-\lambda){\phi'_\ell}^2(x){\rm d}x\, .
 \]
 Then
 \[
S^{\cbH}_{\phi}(\lambda) = \sum_\ell S^{\cbH_\ell}_{\phi_\ell}(\lambda) = 
\sum_\ell \pi\!\int_{\lambda}^{+\infty} (x-\lambda){\phi'_\ell}^2(x){\rm d}x
= \pi\!\int_{\lambda}^{+\infty} (x-\lambda){f}(x){\rm d}x\, ,
 \]
 with $f = \sum_\ell {\phi'_\ell}^2$ almost everywhere. Thus $f\in L^1(\mathbb R_+, dx)$, so $S^{\cbH}_\phi$ is differentiable on $[0, \infty)$ and the derivative
 \[
\frac{d}{d\lambda}S^{\cbH}_\phi(\lambda) = -\pi\!\int_\lambda^{+\infty}f(x){\rm d}x 
 \]
is continuous. Indeed, $\frac{d}{d\lambda}S^{\cbH}_\phi$ is absolutely continuous and
\[
\frac{d^2}{d\lambda^2}S^{\cbH}_\phi(\lambda) = \pi f(\lambda)
\]
almost everywhere as desired. 
\end{proof}
\subsection{The Schr\"odinger representation}\label{App}
In order to analyse the entropy properties associated with any (non smooth) vector and the modular family, we need to study a specific representation. Our results in this section extends those in \cite{L20}, see also \cite{BCD} for further results. 

We realise the Schr\"odinger representation in the Hilbert space $\cH$ is $L^2(\mathbb R_+, p \,dp)$. Given $g\in L^2(\mathbb R_+, p\,dp)$, we extend $g$ to a function $\tilde g\in L^2(\mathbb R, |p|\,dp)$ by setting $\tilde g(-p) = \overline{g(p)}$, $p>0$. Clearly $\tilde g$ is a tempered distribution whose Fourier anti-transform is a real tempered distribution. Call $\mathfrak{L}$ the set of the so obtained distributions:
\[
\mathfrak{L}  = \big\{\phi \in {\cal S}'_{\rm real}(\mathbb R) : \hat \phi|_{[0,\infty)} \in L^2(\mathbb R_+, p\,dp)\big\}\, .
\]
So $\mathfrak{L}$ is a real linear space and the Fourier transform gives a real linear, one-to-one identification of $\mathfrak{L}$ with $\cH$ (as real vector spaces)
\[
\phi\in \frak L \mapsto \hat \phi|_{[0,\infty)} \in \cH \, .
\]
As subspace of $\frak L$, the space $C^\infty_0(\bR)$ of real, compactly supported, smooth functions on $\mathbb R$ 
embeds into into $\cal H$, with dense range.

The translation and dilation one-parameter unitary groups $U$ and $V$ are given as follows and 
define an irreducible, unitary representation of the group $G$ with positive energy.  
We have 
\[
(U(s)\phi)(x)=\phi(x-s) \, , \ \ (V(t)\phi)(x)=e^{-t}\phi(e^t x)\, ,\quad \phi \in \frak L\, ,
\]
so $
 V(t)U(s)=U(e^{-t}s) V(t)$. The anti-unitary reflection $J$ is given by $(J\phi)(x) = \phi(-x)$, $\phi
\in \frak L$. 
 
 For every $\lambda\in\mathbb R$, we set
 \[
 H_\lambda = \big\{\phi\in\frak L : {\rm supp}(\phi) \subset [\lambda,\infty)\big\}\, .
 \] 
 Then $H_\lambda$ is a standard subspace of $\cH$ and $C_0^\infty([0,\infty))$ embeds as a dense, real linear subspace of $H_\lambda$. Clearly $U(\lambda)H_0 = H_\lambda$ and these standard subspaces are associated with the Schr\"odinger half sided modular inclusion $\cbH = (\cH, H_0, H_1)$. 
 
 The following proposition is proved in \cite{L20} (actually, the formula in \cite[Sect. 4]{L20} is more general, for sectors). We begin by giving a simpler proof here. 
\begin{proposition}\label{SS}
With $\cbH$ the Schr\"odinger representation, we have
\begin{equation}\label{S0}
S^{\cbH}_\phi(\lambda) = \pi\!\int_{\lambda}^{+\infty} (x-\lambda)\phi'^2(x){\rm d}x\, ,\quad \phi \in C_0^\infty(\mathbb R)\, ,
\end{equation}
where $\phi'$ is the derivative of $\phi$.
Therefore
\begin{equation}\label{S1}
\frac{d}{d\lambda}S^{\cbH}_\phi(\lambda) = -\pi\!\int_\lambda^{+\infty}\phi'^2(x){\rm d}x \leq 0 \, ,\quad \phi \in C_0^\infty(\mathbb R)\, ,
\end{equation}
\begin{equation}\label{S2}
\frac{d^2}{d\lambda^2}S^{\cbH}_\phi(\lambda) = \pi \phi'^2(\lambda) \geq 0 \, ,\quad \phi \in C_0^\infty(\mathbb R)\, .
\end{equation}
\end{proposition}
\begin{proof}
We have
\[
\Im(\phi, \psi) = \frac{1}2\int \phi'(x)\psi(x)dx\, ,\quad\phi\in C_0^\infty(\mathbb R)\, ,
\]
\[
-(i\log\Delta_{H_0} \phi)(x) = 2\pi x \phi'(x) \, , \quad \phi\in C_0^\infty(\mathbb R)\, ,
\]
\[
P_{H_0}\phi = \phi\chi_{[0,\infty)}\, , \quad \phi\in C_0^\infty(\mathbb R)\, , \ \phi(0) = 0\, .
\]
Therefore, if $\phi\in C_0^\infty(\mathbb R)$, we have
\[
S^{\cbH}_\phi(0) = -\Im(\phi, P_{H_0}i\log \Delta_{H_0}\phi) =
\pi\!\int_{0}^{\infty} x\phi'^2(x){\rm d}x\, ,
\]
that implies \eqref{S0} by translation covariance. The rest of the proposition follows at once. 
\end{proof}
\begin{lemma}\label{ShC}
Theorem \ref{Sl} is true if $\cbH$ is the Schr\"odinger half-sided modular inclusion. 
\end{lemma}
\begin{proof}
With $\phi\in \cH$, we shall show that the function 
$\lambda\in\mathbb R\to S_\phi^{\cbH}(\lambda)$ 
is convex. We may assume that there exists $\lambda_0\in\mathbb R$ with $S_\phi^{\cbH}(\lambda) < \infty$ and show the convexity 
in $(\lambda_0, \infty)$ where 
$S_\phi^{\cbH}(\lambda)$ is finite by the monotonicity of the entropy. By \eqref{S0}, this true if $\phi\in C^\infty_0(\mathbb R)$. 
Set $\lambda_0 = 0$ for simplicity. 

Let $T = \sqrt{|\log\Delta_{H_0}|}E_-$ with $E_-$ the negative $\log\Delta_{H_0}$-spectral projection. By \eqref{Sfin} we have
\[
S_\phi^{\cbH}(0) < \infty \Leftrightarrow \phi\in {D}(T)\, .
\]
Choose a sequence $\phi_n\in C_0^\infty(\mathbb R)$, with supp$(\phi_n)\subset [0,\infty)$, such that $\phi_n \to \phi$ in the graph norm of $T$. 

Now, 
\[
S^{H_0}_\phi\geq
S^{H_\lambda}_\phi = S^{U(\lambda)H_0}_\phi = S^{H_0}_{U(-\lambda)\phi}\, ,\quad \lambda \geq 0\, ,
\]
therefore $U(-\lambda)$ maps $H_0$-finite entropy vectors into $H_0$-finite entropy vectors,
 thus $U(-\lambda){D}(T)\subset {D}(T)$, 
$\lambda \geq0$.
By Lemma \ref{UT}, also $U(-\lambda)\phi_n$ converges to $U(-\lambda)\phi$  in the graph norm of $T$. By \cite{CLR}, we have
\[
S_{\phi_n}^{\cbH}(\lambda)\to S_{\phi}^{\cbH}(\lambda)\, ,\quad \lambda \geq 0\, .
\]
By Prop. \ref{SS}, $S_{\phi_n}^{\cbH}(\lambda)$ is convex. As the pointwise limit of convex functions is convex, we conclude that 
 $S_{\phi}^{\cbH}(\lambda)$ is convex too. 
 \end{proof}
Let $\frak L_0$ be the subspace of $\frak L$ consisting of the tempered distributions $\phi\in\frak L$ such that $\phi'|_{(0,\infty)}$ is a Borel function such that
\[
\int_0^\infty x\phi'(x)^2dx < \infty\, ;
\]
so $\phi'(x)|_{(0,\infty)}\in L^2(\mathbb R_+, xdx)$. We define the real, positive, bilinear form $q$ on $\cH$ with domain $D(q) = \frak L_0$
\[
q(\phi,\psi) = \int_0^\infty x\phi'(x)\psi'(x)dx\, ,\quad \phi,\psi\in D(q)\, .
\]
We also write $q(\phi) = q(\phi,\phi)$ for the associated quadratic form. 
\begin{lemma}\label{qC}
The above form $q$ is closable. 
\end{lemma}
\begin{proof}
We have to show that for any sequence of vectors $\phi_n \in D(q)$ such that $\lim_{n\to\infty} \phi_n = 0$ in $\cH$
and $\lim_{n,m\to\infty}q(\phi_n - \phi_m) = 0$, we have 
$\lim_{n\to\infty}q(\phi_n) = 0$. Namely, if $\phi_n\in\frak L_0$,
\begin{align}
&||\phi_n||^2 =  \int_0^\infty p|\hat \phi_n(p)|^2 dp \to 0, \label{f1}\\
& \int_0^\infty x(\phi'_n(x) - \phi'_m(x))^2 dx \to 0\label{f2}\\
&\Rightarrow \int_0^\infty x\phi'_n(x)^2 dx \to 0\, . \label{f3}
\end{align}
Now, we have by
\eqref{f2} that  $\phi'_n$ is a Cauchy sequence in $L^2(\mathbb R_+, xdx)$ (that we identify with $L^2(\mathbb R, \chi_{\mathbb R_+}(x) x dx)$), so there exists $f\in L^2(\mathbb R_+, xdx)$ such that 
\begin{equation}\label{pL}
\phi'_n\to f \quad {\rm in}\ L^2(\mathbb R_+, xdx)\, .
\end{equation}

As $\phi_n$ is real, thus $|\hat\phi_n|$ is even,
eq. \eqref{f1} means that $\hat\phi_n$ converges to zero in $L^2(\mathbb R, |p|\,dp)$; in particular we have that
$
\int_{-\infty}^\infty p \hat\phi_n(p) \hat \psi(- p)dp \to 0
$
for all $\psi$ in the Schwartz space ${\cal S}(\mathbb R)$; so let $\psi=h'$ with $h\in {\cal S}(\mathbb R)$; then, taking Fourier transforms,
\begin{equation}\label{pL}
\int_{-\infty}^\infty \phi_n(x) h'(x)dx \to 0\, ,
\end{equation}
(integral in a distributional sense), thus
\[
\int_{-\infty}^\infty \phi'_n(x)  h(x)dx \to 0\, , \quad h\in {\cal S}(\mathbb R)\, .
\] 
By \eqref{pL}, we then have 
\[
\int_{-\infty}^\infty f(x)  h(x)dx = 0\, ,
\]
so $f =0$.  Thus \eqref{pL} implies \eqref{f3} and $q$ is closable. 
\end{proof}
\begin{corollary}\label{CorS}
Equation \eqref{S0} holds for all $\phi\in \frak L_0$. For every $\phi\in \cH$, we have
\[
S_\phi^{H_0} = \int_0^\infty x\phi'(x)^2dx
\]
where the right hand side is set equal to $+\infty$ if $\phi'$ does not belong to $L^2(\mathbb R_+, x\,dx)$. 
\end{corollary}
\begin{proof}
Let $\cD$ be the set of vectors $\phi\in\cH$ such that $S_\phi^H < \infty$.
With $h, k\in \cD$, we set
\[
S(h, k) = -\Im(h, P_H i\log\Delta_H k)   \, ,
\]
which is well defined by the real polarisation identity. Then $S(\cdot,\cdot)$ is a bilinear form on $\cD$ whose associated quadratic form is lower semicontinuous, hence closable. Moreover $C^\infty_0 (\mathbb R)$ is a form core for $S(\cdot,\cdot)$;
this follows because $C^\infty_0 (\mathbb R)$ is a core for the operator $\sqrt{\log\Delta_H} E_-$, see \cite[Prop. 2.4]{CLR} and  preceding discussion. 
Since $q$ and $S$ coincide on $\cD$, they must agree on the form closure $\cD$ by Lemma \ref{qC}. 
\end{proof}
\begin{theorem}\label{STh}
Proposition \ref{SS} holds true for all $\phi\in\cH$. 
\end{theorem}
\begin{proof}
By Cor. \ref{CorS}, we immediately get
\[
S_\phi^{H_\lambda} = \int_\lambda^\infty (x-\lambda)\phi'(x)^2dx\, ,
\]
hence \eqref{S1}, \eqref{S2} follow by differentiation and clearly entail the theorem. 
\end{proof}
The following elementary lemma was needed.  
\begin{lemma}\label{UT}
Let $\cH$ be a Hilbert space, $T : {D}(T)\subset \cH\to \cH$ a closed linear operator. If $U\in B(\cH)$ maps ${D}(T)$ into itself, then $U|_{{D}(T)}$ is bounded operator on ${D}(T)$ with the graph norm of $T$. 
\end{lemma}
\begin{proof}
By the closed graph theorem, it suffices to show that $U|_{{D}(T)}$ is closable with respect to the graph norm, namely that
\[
\xi_n,\eta \in {D}(T), \ ||\xi_n||+ ||T\xi_n||\to 0\, ,
 \  ||U\xi_n- \eta||+ ||TU\xi_n- T\eta||\to 0 \Rightarrow \eta = 0\, ,
\]
that holds true because $U$ is bounded, so $\xi_n\to 0$ implies $U\xi_n \to \eta =0$. 
\end{proof}

\subsection{Algebras and states associated with a one-particle structure}\label{QF1}
We now recall the Weyl algebra and the one-particle structure associated with a quasi-free state. 

Let $S$ be a symplectic space, that is $S$ a real linear space and $\sigma$ is a non-degenerate symplectic form on $S$; thus $\sigma$ is a real, bilinear, anti-symmetric form on $S\times S$. 
Given a real scalar product $\mu$  on $S$ satisfying the inequality 
\[
 \sigma(f_1,f_2)^2\leq   \mu(f_1,f_1)\cdot \mu(f_2,f_2) \ , \qquad f_1,f_2\in S\, , 
\]
 a \emph{one-particle structure} $(\cH_\mu,\kappa_\mu)$  on 
$S$ \cite{KayWa91} is given by a complex Hilbert space $\cH_\mu$ and  a real linear mapping $\kappa_\mu: S\to\cH_\mu$ satisfying
\begin{enumerate}
\item $\kappa_\mu(S)+i\kappa_\mu(S)$ is dense in $\cH_\mu$\, ,
\item $\Re(\kappa_\mu(f_1),\kappa(f_2))= \mu(f_1,f_2)$ and $\Im(\kappa(f_1),\kappa(f_2))= \sigma(f_1,f_2)$\, .
\end{enumerate} 
A one-particle structure is unique, modulo unitary equivalence \cite{Kay1985u}. 

Now, given a Hilbert space $\cH$,
we  consider  the Bose-Fock space 
$\Gamma(\cH):=\bigoplus^\infty_{k=0} \cH^{\otimes^k_s}$,  
where $\cH_0\equiv\mathbb C\xi$ is the one-dimensional Hilbert space of a unit vector $\xi$, the \emph{vacuum vector}, and $\cH^{\otimes^n_s}$ is the symmetric $n$-fold tensor product of $\cH$. To any $\phi\in\cH$
there corresponds a \emph{coherent vector} $e^\phi := \bigoplus^\infty_{n=0} \frac{1}{\sqrt{n!}} \phi^{\otimes^n_s}$ on $\Gamma(\cH)$,
where the zeroth component of $e^\phi$ is $\xi$. Coherent vectors form  a total family of linearly independent vectors of $\Gamma(\cH)$, whose scalar product verifies
$(e^\phi,e^\psi)= e^{(\phi,\psi)}$, for every  $\phi,\psi\in\cH$.

The {\it Weyl unitaries} on  $\Gamma(\cH)$ are first defined on coherent vectors by
\[
 W(\psi) e^\phi:=  e^{\psi +\phi}\cdot e^{-\frac{1}{2}(\psi,\psi)-(\psi,\phi)} \ , \qquad \psi,\phi\in\cH\ ,
\]
and then extended by linearity and density to all the Fock space. They satisfy  
the \emph{Weyl commutation relations}
\[
W(\psi)\, W(\phi)= e^{-i\Im(\psi,\phi)} W(\psi+\phi) \ ,  \qquad \psi,\phi\in\cH \ . 
\]
The \emph{Weyl algebra} $\cA(S)$ associated with the symplectic space $(S,\sigma)$ can be defined as the abstract $\rC^*-algebra$ generated the Weyl relations, but we shall directly deal with its representation associated with a one-particle structure. Given a one-particle structure $(\cH_\mu,\kappa_\mu)$ as above, let $\cA(S)$ be the $\rC^*$-algebra obtained as the norm closure of the $*$-algebra generated by the Weyl operators $W_\mu(\kappa_\mu(f))$ as $f$ varies in $S$, where $W_\mu$ denotes the Weyl unitaries on $\Gamma(\cH_\mu)$. 

The \emph{(quasi-free) state} $\varphi_\mu$ of $\cA(S)$ associated with $\mu$ is determined by 
\[
\omega_\mu(W_\mu(\kappa_\mu(f))) =  (\xi_\mu, W_\mu(\kappa_\mu(f))\xi_\mu) = e^{-\frac12 ||\kappa_\mu(f)||^2}= e^{-\frac12 \mu(f,f)}\ , \qquad f\in S\, ,
\]
with $\xi_\mu$ the vacuum vector of $\Gamma(\cH_\mu)$. Clearly, it is a normal state of the von Neumann algebra generated by $\cA(S)$. 

Note that any real linear, invertible map  $T: S\to S$, that preserves $\sigma$ and $\mu$,
promotes to a unitary on $\cH_\mu$, hence to a unitary  $u_{\mu,T}$ on the Fock space $\Gamma(\cH_\mu)$ which satisfies 
\[
u_{\mu,T}\,  W_\mu(\kappa_\mu(f)) \, u^*_{\mu,T}= W_\mu(\kappa_\mu(T f))\, ;
\]
so $u_{\mu,T}$ implements a vacuum preserving automorphism of $\cA(S)$. 

\section{Relative entropy and globally hyperbolic spacetimes}
\label{relspacetime}
We now make our analysis in the context of globally hyperbolic spacetimes admitting 
a Killing flow that is timelike and complete within certain causally convex subregions.  
We first recall some basic facts and then introduce 
admissible regions so to have, at the second quantisation level,  
half-sided modular inclusions of von Neumann algebras. 
We then show that, in the real scalar Klein-Gordon field case, the one-parameter family of relative entropy associated with a coherent state is convex. 

We start with a brief description of the causal structure, we refer to \cite{DMP17} for a detailed account \cite{EH73,Wa84,MTW} and for standard textbooks. 
Let $\cM$ be a 4-dimensional \emph{spacetime} i.e.\ a connected, time-oriented, Lorentzian manifold. For any subset $A\subset\cM$,   
the symbols $I^+(A)$, $J^+(A)$ and $D^+(A)$ denote  the \emph{chronological future},  the \emph{causal  future} 
and the \emph{future domain of dependence} of $A$. The corresponding past notions will be denoted
by replacing $+$ with $-$. 

A subset $A$ of $\cM$ is \emph{achronal}  $I^+(A)\cap A =\emptyset$. The \emph{edge} of an achronal set $A$ 
is the set of the points $p$ in the closure $\overline{A}$ such that for any neighbourhood $U$ of $p$
there are $p_\pm\in I^\pm(p,U)$  joined by a timelike curve contained in $U$ not intersecting $A$; here $I^\pm(p,U)$ is the set of points joined by a \emph{future/past-directed} (f-d/p-d) timelike curve starting from $p$ and contained in $U$.  
An achronal subset $S$ such that  $D^+(S)\cup D^-(S)=\cM$ is called a \emph{Cauchy surface of $\cM$}. 
A spacetime $\cM$ is \emph{globally hyperbolic} if one of the following equivalent properties is verified:

$(i)$  $\cM$ admits a Cauchy surface;

$(ii)$ $\cM$ admits a spacelike Cauchy surface;

$(iii)$ the collection  $\{I^+(p)\cap I^-(q) \ , \ p,q\in\cM\}$ is  a base of neighbourhood 
for the topology of $\cM$ and  the sets $J^+(p)\cap J^-(q)$ are compact for any pair of points $p,q\in \cM$.
  
\noindent 
A subset $A\subset \cM$ is \emph{causally convex} whenever $J^+(A)\cap J^-(A)=A$. 
The \emph{causal convex hull} $\ch(S)$ of the set $S$ is the smallest causally convex set containing $S$. Clearly  
\[
\ch(S)= J^+(S)\cap J^-(S)  \ , \qquad S\subset\cM.
\]
By property $(iii)$, any open causally convex subset of a globally hyperbolic spacetime $\cM$ is itself a globally hyperbolic spacetime.

\subsection{Wedges and strips}  
\label{relspacetime:B}
 Generalisation of wedge shaped regions of the Minkowski spacetime to curved spacetimes have been studied by several authors, see \cite{DLM} and references therein. 
Here, we give notions more suited to our aims. 

Within this section  we consider a connected globally hyperbolic spacetime $\cM$ endowed 
with a Killing flow $\Lambda$, i.e. a one-parameter group of isometries.
\begin{definition}
\label{wedge}
A \emph{wedge} of $\cM$ associated with the Killing flow $\Lambda$ is an open connected, causally convex subset $\cW$ of $\cM$ such that:

$(i)$ $\Lambda_s: \cW\to \cW$ is a diffeomorphism, $s\in \mathbb R$,

$(ii)$ $\Lambda$ is timelike, time oriented and complete within $\cW$,

\noindent
where time oriented means that the Killing vector field generating $\Lambda$ is either 
f-d for all points of $\cW$ or p-d for all points of $\cW$.
\end{definition} 
\noindent
Note that, since a wedge $\cW$ is open and causally convex, $\cW$ is itself a globally hyperbolic spacetime 
with the induced Lorentzian structure; as $\Lambda$ is assumed to be timelike complete, $\cW$ is indeed a globally hyperbolic \emph{stationary} spacetime.  

We set 
\[
 \rO_A:=\big\{\Lambda_s (p) \ : \ p\in A, \ s\in\bR \big\} \, , 
\]
for the \emph{orbit} of the set $A\subset \cW$. 
We also consider the \emph{positive/negative} half orbit  
$\rO^\pm_A:=\{\Lambda_s (p) \ : \ p\in\ A, \ \pm s> 0\}$. 
\begin{lemma}
\label{equiv-wedge}
Let $\cW$ be a wedge of $\cM$ associated with $\Lambda$. For any $p\in \cW$, 
\[
 \cW= \ch(\rO_p)=I^-(\rO_p)\cap I^+(\rO_p) \, . 
\]
\end{lemma} 
\begin{proof}  
Without loss of generality we assume that $\Lambda$ is f-d. 
We prove that $I^-(\rO_p)\cap I^+(\rO_p)=\ch(\rO_p)$. 
To this end  note that $I^+(\rO_p)\cap I^-(\rO_p)$ is causally convex 
(this easily follows by the properties of causal sets).   
Hence  $\ch(\rO_p)\subset I^+(\rO_p)\cap I^-(\rO_p)$. Conversely,
given $x\in I^+(\rO_p)\cap I^-(\rO_p)$ there are $s_1,s_2\in\bR$ such that\ $s_1<s_2$ 
and $x\in I^+(\Lambda_{s_1}(p))\cap I^-(\Lambda_{s_2}(p))$. Taking $s_3<s_1<s_2<s_4$,
we have that 
\[
x\in I^+(\Lambda_{s_1}(p))\cap I^-(\Lambda_{s_2}(p))\subset J^+(\Lambda_{s_3}(p))\cap J^-(\Lambda_{s_4}(p))
\subset \ch(\rO_p) \, .
\]
So $I^-(\rO_p)\cap I^+(\rO_p)=\ch(\rO_p)$ and $I^-(\rO_p)\cap I^+(\rO_p)\subset \cW$ because $\cW$ is causally convex. 
On the other hand, as $\cW$ is a globally hyperbolic stationary spacetime with respect to $\Lambda$, we have 
\[
 \cW= I^-(\rO_p, \cW)= I^+(\rO_p, \cW) \ .
\]
In fact, taking  a smooth spacelike Cauchy surface $\cC$ containing $p$,  
the mapping $\bR\times \cC \ni(s,x)\mapsto \Lambda_s(x)\in \cW $ 
is an isometry if $\bR\times \cC$ is equipped with the induced metric 
\[
g_{(t,x)}= - \beta(x)dt^2 + 2\omega_x dt  + h_x \ , \qquad (t,x)\in \bR\times \cC \ , 
\] 
where $\omega$ is a 1-form,  $\beta$ is a smooth strictly positive function, 
$h$ is a Riemannian metric over $\cC$ \cite[Theorem 2.3]{CFS}.
Then, for any point $q$ in $\cC$ 
take a spacelike curve $\alpha:[0,1]\to \cC$ with $\alpha(0)=q$ and $\alpha(1)=p$. 
As $\beta$ is positive and depends only on $\cC$, and as $\Lambda$ is complete, by continuity and compactness
of $[0,1]$ one can always find a smooth increasing function $f:[0,1]\to [0,+\infty)$ such that 
the curve $\rho(t):=(f(t),\alpha(t))$ is f-d timelike.
Clearly  $\rho(0)=(0,q)$ and $\rho(1)=\Lambda_{f(1)}(p)=
(f(1),p)$. So  $q\in I^-(\rO_p,\cW)$ and as this holds for any point of $\cC$ we have that 
$I^-(\rO_p,\cW)=\cW$. Therefore, 
\[
 \cW= I^-(\rO_p, \cW)\cap I^+(\rO_p, \cW)\subset I^-(\rO_p)\cap I^+(\rO_p)\subset \cW
\]
completing the proof.   
\end{proof} 
Given a wedge $\cW$ of $\cM$ associated with $\Lambda$, we say that an  open connected 
subset $\cV\subset \cW$ is (positively) \emph{half-invariant} w.r.t.\ $\Lambda$ 
whenever 
\[
 \Lambda_s(\cV)\subset \cV \, , \qquad \forall\, s>0 \, .
\]   
Note that, given a $\Lambda$ half-invariant  subset $\cV$,  then  
$\Lambda_s(\cV)$ is  half-invariant  too, $s >0$. 

We are going to consider two kind of  half-invariant  subregions of a wedge $\cW$ of $\cM$. 
Let  $\tau$ be an isometry of $\cM$. The image $\tau(\cW)\equiv\cW^\tau$ is a wedge with the respect 
to the one-parameter group 
$ \mathrm{ad}_\tau(\Lambda):= \tau\circ \Lambda\circ \tau^{-1}$; if $
\cW^\tau \subset \cW
$, we shall say that $\cW^\tau$ is a \emph{subwedge} of $\cW$.  

Now, $\cW^\tau$ is $\Lambda$ half-invariant whenever $\Lambda_s(\tau(\cW))\subset \tau(\cW)$ for $s>0$. By
Lemma \eqref{equiv-wedge}
\[
\cW^\tau=\ch\big(\{\mathrm{ad}_\tau(\Lambda_s)(\tau(p)) , \ s\in\bR \}\big)= \ch\big(\{\tau\circ \Lambda_s(p), \ s\in\bR \}\big)= 
\tau\big(\ch(\rO_p)\big) 
\]
for any $p\in \cW$, so we have
\[
(\Lambda_t\circ \tau) (\rO_{p}) \subset\tau(\ch(\rO_p)) \iff \mathrm{ad}_{\tau^{-1}}\circ\Lambda_t (\rO_{p}) \subset\ch(\rO_p)   \, . 
\]
Even though the spacetime has no symmetries besides $\Lambda$, there are always proper half-invariant subregions of $\cW$ that, in general, are neither wedges nor causally convex. We now give an interesting family of such regions.
\begin{definition}
\label{strip}
Given a connected achronal subset $A$ of $\cW$ with $A\cap\mathrm{edge}(A)=\emptyset$, the half-orbit  
$\rO^+_A$
is called the (positive) \emph{strip} generated by $A$. 
\end{definition}
\begin{lemma}
\label{strip:open}
Let $A\subset \cW$ be a connected achronal such that $A\cap \mathrm{edge}(A)=\emptyset$. 
Then $\rO^+(A)$
is an open connected half-invariant subset of $\cW$.
\end{lemma}
\begin{proof}
We  assume that $\Lambda$ is f-d and  $\rO^+_A$ is clearly connected. Furthermore as $A$ is achronal 
and $\Lambda$ timelike  then $A\cap \Lambda_s(A)=\emptyset$ for any $s>0$, so  
\[
\Lambda_{s_1}(A)\cap \Lambda_{s_2}(A)=\emptyset \, , \ \  \ \ 0<s_1 <s_2 \, . 
\]   
So $\rO^+_A$ is half-invariant w.r.t.\ $\Lambda$. What remains to be shown is that $\rO^+_A$ is open. Given  $p\in\rO^+_A$, then $p\in\Lambda_{s_*}(A)\equiv A_{s_*}$ for some $s_*>0$. 
$A_{s_*}$ is achronal and $A_{s_*}\cap \mathrm{edge}(A_{s_*})=\emptyset$.  
The latter, in particular, implies that we can find a relatively compact 
subset $U$ s.t.\ $p\in U\subset \overline{U}\subset A_{s_*}$ and  
\[
U\cap \mathrm{edge}(U)=\emptyset \, .
\]
By \cite{BS06} there is an achronal Cauchy surface  $\cC$ containing the closure $\overline{U}$ of $U$.  Then  as $\Lambda$ is timelike and complete, the mapping 
\[
\bR\times \cC \ni (s,x)\mapsto \Lambda_s(x)\in \cM 
\]
is an homeomorphism \cite[Lemma 2]{BS03}: it is continuous and bijective because for any point $p\in\cM$ the timelike curve $\Lambda_t(p)$ meets $\cC$ exaclty once. The invariance of domain theorem implies that this mapping is a homeomorphism. 
 
As $U\cap \mathrm{edge}(U)=\emptyset$, we have that $U$ is open in $\cC$. It is enough to observe that 
the limit points of any sequence $x_n\in \cC\setminus U$ belong to $\cC\setminus U$ because otherwise \cite[Proposition 2.140]{Min2019} 
they would be edge points of $U$ (the sequence $x_n$ is in the complement of $I^+(U)\cup I^-(U)\cup U$ ). 

In conclusion, for $\eps>0$ and small enough, the set   
$\{ \Lambda_s(U) \ : \ s\in (s_*-\eps, s_*+\eps)\}$ is an open subset of $\rO^+_A$
containing $p$, and this proves that $\rO^+_A$ is open.
\end{proof}
Note that the causal complement of a strip is empty in the wedge as follows by applying the same reasoning 
of the proof  the Lemma \ref{equiv-wedge}. 

\subsubsection{Examples of half invariant regions}
\label{relspacetime:C}

\textbf{Minkowski spacetime.}\label{minkowski1}
Let $\cM = \mathbb R^{4}$ be the $4$-dimensional Minkowski spacetime with  
Poincar\'e symmetry group.
With $\Lambda$ the Killing flow  given  the pure Lorentz transformations in the $x_1$ direction,
a wedge with respect to $\Lambda$ is the usual right wedge $\cW_0 = \{x\in\mathbb R^{4}: x_1 > |x_0|\}$. 

 Half invariant subregions of $\cW_0$ are the lightlike translations 
$\cW_\lambda:=(\lambda,\lambda,0,0)+\cW_{0}$ with $\lambda>0$. Note that $\cW_\lambda$ is indeed a strip
generated by 
\[
 A_\lambda:=\{(x_0,x_1,y)\in \cM \ | \  x_0 = -x_1+\lambda \ , \ x_1>\lambda \ , \ y\in\bR^2\big\}  \, .  
\]
which is a connected achronal subset of $\cW_0$ with $A_{\lambda}\cap\mathrm{edge}(A_{\lambda})=\emptyset$. 

Other examples include the \emph{deformed wedges} of $\cW_0$ \cite{BFLW15}, also discussed in \cite{L20, MTW21}. Any non-negative smooth function $f(y)\ne 0$, where $y=(x_2,x_3)$, yields a deformation $\cW_f$ of the right-wedge $\cW_0$, defined as
\begin{equation}
\label{deformation}
\cW_f:=\big\{(x_0,x_1,y)\in\cM \ , \ |x_0-f(y)|< x_1-f(y) \big\} \, . 
\end{equation}
Although not a wedge, as it lacks causal convexity and $\Lambda$ invariance, $\cW_f$ is a half-invariant subregion of $\cW_0$. The corresponding lightlike translation $\cW_{f+\lambda} = (\lambda,\lambda,0,0)+\cW_{f}$, with $\lambda>0$, is also half-invariant. This is generated by
\begin{equation}
\label{strip-deformed}
A_{f+\lambda}:=\big\{(x_0,x_1,y)\in \cM \ | \ x_0 = -x_1+2(f(y)+\lambda) \ , \ x_1> f(y)+\lambda \ , \ y\in\bR^2\big\} \, .
\end{equation}
While $A_{f+\lambda}$ is a connected smooth surface, unlike $A_{\lambda}$, it is not achronal unless $f$ is constant. Consequently, $\cW_{f+\lambda}$ is not a strip in general.\medskip\noindent

{\bf Rindler spacetime.} 
\label{rindler}
The Rindler spacetime is a wedge itself: it may be considered as the above standard wedge of the Minkowski spacetime with hyperbolic coordinates \cite{CLR}. The Lorentz boosts of Minkowski correspond to time translation in Rindler spacetime giving a time-like Killing flow $\Lambda$ that becomes light-like on the boundary of the wedge (horizons). Strips are defined as in the Minkowski case.  

\medskip\noindent
{\bf  Kruskal spacetime.}\label{kruskal}
The Kruskal spacetime $\cM$ is a 4-dimensional globally hyperbolic spacetime  arising  as the maximal analytic extension of the \emph{Schwarzschild spacetime}. 
We stare here few  basic properties, for further details see e.g. the  references at the beginning of this section. 

In the Kruskal-Szekeres coordinates the metric is defined on $\{ 
(t,x) \in\bR^2 : x^2-t^2 <-1\}\times \bS^2$ as
\[
ds^2 =  \frac{32M^3}{r}\, e^{-\frac{r}{2M}}\, (-dt^2 + dx^2)+ r^2 d\Omega^2
\]
where $d\Omega^2$ the area element of the unit 2-sphere $\bS^2$, $M>0$ is the black hole mass, and  $r\in (0,+\infty)$ is the  Schwarzschild radius  implicitly related to $t,x$ by  
\[
x^2-t^2= e^{\frac{r}{2M}}\left(\frac{r}{2M}-1\right) \, . 
\]
The metric has a ``physical" singularity at $x^2-t^2=-1\iff r=0$.   

There are four Killing vector fields generating the symmetry group: 
the ones associated with the spatial rotation group $SO(3)$, and the time translation Killing flow that acts on the $(t,x)$ component by
\[
\Lambda_s =\begin{pmatrix}
	\cosh (s/4M) & \sinh(s/4M) \\
	\sinh(s/4M) & \cosh(s/4M) \\
\end{pmatrix} 
\]
and trivially on the $\bS^2$ component. $\Lambda$ is timelike in the right wedge  $\cW_+:=\{t^2-x^2>0 \ , x>0\}\times \bS^2$, corresponding to the external region of the Schwarzschild black hole, and in the left wedge $\cW_-:=\{t^2-x^2>0 \ , x<0\}\times\bS^2$. 

Just like in the case of the Minkowski spacetime the lightlike translations  $\cW_{+,\lambda} = \cW_+ + (\lambda, \lambda)$, with $\lambda > 0$ (translation on the $\bR^2$ component) turn out to be strips. $\cW_{+,\lambda}$ represents 
the half-orbit of the translated right-past horizon $h_{\lambda}:= \{x^2-t^2=0, \ x>0, \ t<0\}\times \bS^2+(\lambda,\lambda)$. 
Other half-invariant subregions of $\cW_+$ are the deformations defined as 
\[
\cW_{+,f} := \big\{(t,x,\Omega)\in\cM \, | \, |t-f(\Omega)|< x-f(\Omega) \big\} \ ,
\] 
for any non negative smooth function $f\ne 0$ on $\bS^2$, along with their corresponding $\lambda$-translated 
versions $\cW_{+,f}+(\lambda,\lambda)$ for any $\lambda>0$. The latter, much like in the Minkowski case, are not strips in general; they are generated by
\begin{equation}
\label{strip-deformed2}
h_{f+\lambda}:=\big\{(t,x,\Omega)\in \cM \, | \, t = - x +2(f(\Omega)+\lambda) , \,  x> f(\Omega)+\lambda \big\} \, ,
\end{equation}
which are not achronal subsets of $\cM$ unless $f$ is constant.

\medskip\noindent
{\bf  Further examples.}\label{further}
As is known, wedge regions of the $d$-dimensional de Sitter spacetime $dS^d$ are naturally defined by the embedding of $dS^d$ into in $\R^{1 + d}$; strips can be similarly defined.  
Concerning the Minkowski spacetime $\cM$ with time translation Killing flow $T$, then the forward light cone is a $T$ half-invariant subregion, indeed a strip of the wedge $\cM$; in the one-dimensional case, a positive half-line is a half invariant subregion of $\mathbb R$. KMS states in these contexts are studied respectively in \cite{RST,BY99} and \cite{CLTW1,CLTW2,L01}. 

We finally point out that the relative entropy concept in this paper is also applicable 
to spacetimes describing non-stationary black holes, as recently discussed in \cite{KPV21}.

\subsection{Entropy and  Klein--Gordon field on a globally hyperbolic spacetime}
\label{relspacetime:E} 
We now consider \emph{Weyl quantisation} of the Klein-Gordon free scalar field on a globally hyperbolic spacetime 
$\cM$. 

The Klein-Gordon operator is $-\Box +m^2$, where $\Box$ is the D'Alembertian associated with the spacetime metric tensor and $m$ is the mass. 
We consider the space of compactly supported, smooth functions $C^\infty_0(\cM,\bR)$ equipped with the symplectic form 
\[
\sigma(f_1, f_2):=\frac12\int_{\cM} f_1 E(f_2) \, d\nu  \, , \qquad f_1,f_2\in  C^\infty_0(M,\bR) \, ,
\]
where $E$ is the  advanced-minus-retarded propagator of  $-\Box +m^2$ and 
$\nu$ is the metric-induced volume measure on $\cM$.

The symplectic form $\si$ annihilates
on the image of the Klein-Gordon operator. So we consider 
the quotient $S(\cM):=C^\infty_0(\cM,\bR)/ [(-\Box + m^2)(C^\infty_0(\cM,\bR))]$ where $\si$  is non degenerate. For simplicity, 
in the following we omit the equivalence class symbol when considering elements of $S(\cM)$. 

The Weyl algebra of the Klein-Gordon field is the Weyl $\rC^*$-algebra associated with the symplectic space $(S(\cM),\sigma)$ as in Section \ref{QF1}. 

We consider a \emph{one-particle structure} $(\cH_\mu,\kappa_\mu)$, where $\mu$ is a real scalar product on $S(\cM)$ as in Section \ref{QF1}. 
For simplicity, we set
\[
W_\mu(f):= W_\mu\big(\kappa_\mu(f)\big) \, , \qquad f\in C^\infty_0(\cM)\, .
\]
and note that the Weyl unitaries satisfy 
the \emph{Klein-Gordon equation} 
\[
W_\mu\big((-\Box + m^2)f\big)=1 \,  , \qquad f\in C^\infty_0(\cM) \, . 
\]
Given a globally hyperbolic spacetime with a Killing flow $\Lambda$, let $\varphi_\mu$ be a quasi-free state,
defined by a one-particle structure $(\cH_\mu,\kappa_\mu)$  which is left invariant by $\Lambda$. 
For any region $\cO$ with nonempty interior of $\cM$, we consider the real subspace
\[
 H_\mu(\cO):= \big\{\kappa_\mu(f) \, , \, \text{$f$ compactly supported in $\cO$}\big\}^{-}\subset \cH_\mu
\]
and the von Neumann algebra
\[
\cR_\mu(\cO):=\{  W_\mu(f)  \, , \, \text{$f$ compactly supported in $\cO$}\}'' \, .
\] 
The correspondence $\cO\mapsto\cR_\mu(\cO)$ satisfies 
\begin{itemize}
\item  $\cO_1\subset \cO_2 \ \Rightarrow \ \cR_\mu(\cO_1)\subset \cR_\mu(\cO_2)$  \emph{(isotony)};
\item $\cO_1\perp \cO_2 \ \Rightarrow \ [\cR_\mu(\cO_1),\cR_\mu(\cO_2)]=0$, \emph{(causality)};
\item  $u_{\mu,T}\,  \cR_\mu(\cO) u_{\mu,T}^* =\cR_\mu(\tau\cO)$ if $\tau$ is an isometry of $\cM$ leaving $\mu$  invariant and $T$ is the associated symplectic transformation (see Sect. \ref{QF1})
\emph{(covariance)}; 
\end{itemize}
here $\cO_1\perp \cO_2$ means that the closure of $\cO_1$ and $\cO_2$ are causal disjoint. 
We note that this representation is  irreducible when the real subspace $\kappa_\mu(S(\cM))$ is dense in $\cH_\mu$. Other properties of the Weyl algebra of the Klein-Gordon field can be found in \cite{Ver1997}.

We now assume that $(\cM,\Lambda)$ is \emph{stationary}, i.e.\ $\Lambda$ is timelike complete. In this case the Weyl algebra of the Klein-Gordon field admits, for any $\beta>0$ and $m>0$, a unique quasi-free $\beta$-KMS Hadamard 
state $\cite{San13}$ with respect to the timelike Killing flow $\Lambda$. In the following, we  fix $\beta>0, m>0$; the one-particle structure is the one associated with  this quasi-free $\beta$-KMS state and we set $\mu = \beta$.
\begin{proposition}
\label{convexity:stationary}
Let $(\cM,\Lambda)$ be a stationary spacetime and let $\varphi_\beta$  
be a $\beta$-KMS quasi-free Hadamard state of  the Weyl algebra. Then, if  $\cV$ is any non-empty  half-invariant region of $\cM$,
the triple $\boldsymbol{\cH_\beta}:=\big(\cH_\beta,H_\beta(\cM), H_{\beta}(\cV)\big)$ 
is a half-sided modular inclusion;  
the  relative entropy  function 
\[
S^{\cR_\beta(\cV_\lambda)}(\varphi_\phi|\!| \varphi_\beta) = S^{\boldsymbol{\cH_{\beta}}}_{\phi}(\lambda) \ ,
\]
is convex with respect to $\lambda$. 
\end{proposition}
\begin{proof}
Let $u_{\beta,s}$ be the unitary on $\Gamma(\cH_\beta)$ associated with $\Lambda_s$. 
The $\beta$-KMS condition implies that 
\begin{equation}\label{*}
\mathrm{ad}\big(u_{\beta, s}\big) = \mathrm{ad}\big(\Delta^{-is}_{\cR_\beta(\cM)}\big) \, ,
\end{equation}
i.e. the one-parameter automorphism group of the von Neumann algebra $\cR_\beta(\cM):=\cA(\cM)''$ on the Fock space is the modular group  ad$\Delta^{-is}_{\cR_\beta(\cM)}$.

The von Neumann algebra associated with a non empty, open, relatively compact region satisfies the Reeh-Schlieder property \cite{Stro2000}. 
So, the corresponding real Hilbert space $H_\beta(\cO)$ is a standard subspace of $\cH_\beta$ and  
\begin{equation}\label{**}
\Delta_{\cR_\beta(\cO)} = \Gamma(\Delta_{H_\beta(\cO)}) \ \ , \ \ J_{\cR_\beta(\cO)} = \Gamma(J_{H_\beta(\cO)})   \ , \qquad \cO\subset \cM \, ,
\end{equation}
where $\Delta_{H_\beta(\cO)}$ and $J_{H_\beta(\cO)}$ is the modular operator of $H_\beta(\cO)$. Now,  since $\cV$ is half invariant,  \eqref{*} and \eqref{**} give
\[
\Delta^{-is}_{H_{\beta,0}}H_{\beta}(\cV)=  H_\beta(\Lambda_s(\cV))\subset H_{\beta}(\cV) \, ,\qquad s> 0\, ,
\]
that implies the triple $\boldsymbol{\cH_\beta}:=\big(\cH_\beta,H_\beta(\cM), H_{\beta}(\cV)\big)$ to be a half-sided modular inclusion. Hence, by Proposition \ref{equiv-entropy}, we have that 
\[
S^{\cR_\beta(\cV_\lambda)}(\varphi_\phi|\!| \varphi_\beta)=S^{\boldsymbol{\cH_{\beta}}}_{\phi}(\lambda)  \ .
\]
The convexity then follows by Theorem \ref{Sl}. 
\end{proof}

The above results can be generalized in the context of a globally hyperbolic spacetime  $\cM$ with a wedge $\cW$ with respect to the Killing flow $\Lambda$.     
Being $(\cW,\Lambda)$ globally hyperbolic and stationary,
the corresponding Weyl algebra has  $\beta$-KMS state $\varphi_\beta$. 
If $\varphi_\beta$ admits an extension to a quasi-free $\Lambda$-invariant state $\varphi_o$ 
(the ``vacuum state") of the Weyl algebra of  
$\cM$ such that  the vacuum vector $\xi_o$ is cyclic for the algebra $\cR_o(\cW)$,
then we may consider the unitary $U_{o,\beta}$ defined by
\begin{equation}
\label{extension}
U_{o,\beta}  W_\beta(f)\xi_\beta := W_o(f)\xi_o \ , \qquad f\in C^\infty_0(\cW) \, ;
\end{equation}
then the results of Propositions \ref{convexity:stationary} easily extend  
to the vacuum representation of the Weyl algebra of $\cM$. 
\begin{corollary}
\label{convexity:general}
Let $\cM$ be a globally hyperbolic spacetime with a wedge $\cW$
with respect a Killing flow $\Lambda$. Let $\varphi_\beta$  be a KMS quasi-free state of the Weyl algebra of $\cW$ which admits an extension to a quasi-free state $\varphi_o$ of the Weyl algebra of $\cM$ satisfying \eqref{extension}.
Then, for any  half-invariant region $\cV$ of $\cW$ the triple $\boldsymbol{\cH_o}:=(\cH_o,H_o(\cW), H_{o}(\cV))$ 
is a half-sided modular inclusion.
The  relative entropy  function $S^{\boldsymbol{\cH_o}}_{\phi}(\lambda) = S^{\cR_o(\cV_\lambda)}(\varphi_\phi|\!| \varphi_o)$ is convex with respect to $\lambda$.
\end{corollary}
We point out that the extension condition \eqref{extension} is verified in meaningful examples. 
In the Minkowski spacetime, it holds because of the Bisognano-Wichmann Theorem \cite{BW}.  In the de Sitter spacetime, the de Sitter vacuum \cite{GH77, BB99} is $\beta_{GH}$--KMS for every wedge $\cW$ with respect to $\Lambda_\cW$, where $\beta_{GH}$ is the Gibbons-Hawking inverse temperature. 

A further relevant example is provided by the Kruskal spacetime; the analogy of the Bisognano-Wichmann setting of QFT for wedge regions in Minkowski spacetime and the for outer Schwarzschild region of Kruskal spacetime was noticed in \cite{Sew82}.
The Hartle-Hawking-Israel state satisfies the Bisognano-Wichmann property, this result is essentially contained in  \cite{San15}, using the analysis in \cite{Kay1985}, as shown in the following.   
\begin{lemma}
Let $\varphi_o$ be the Hartle-Hawking-Israel state of the scalar Klein-Gordon field on the Kruskal spacetime. Then

$(i)$ the restriction of $\varphi_o$ to $\cA(\cW_+)$ is $\beta_H$-KMS with respect to $\Lambda$, with $\beta_H$ the Hawking inverse temperature; 

$(ii)$ $\varphi_o$ satisfies the Reeh Schlieder cyclicity property with respect to open, non-empty, relatively compact subregions in the double wedge $\cW_- \cup \cW_+$;

$(iii)$ $\varphi_o$ satisfies the Bisognano-Wichmann property 
\[
 \mathrm{ad}(J_{\cR_o(\cW_+)}) \cR_o(\cO)= \cR_o(\iota\cO) \, ,  \ \mathrm{ad}\Big(\Delta^{-is}_{\cR_o(\cW_+)}\Big)\cR_o(\cO) = \cR_o(\Lambda_{2\pi s}\cO) \, , 
 \quad \cO\subset \cW_+\, ,
\] 
where $\iota$ is the reflection symmetry of $\cM$ on the $\mathbb R^2$ component: $\iota(t,x)=(-t,-x)$. 
\end{lemma}
\begin{proof}
Sanders has shown in \cite[Theorem 5.3]{San15} that the HHI state $\varphi_o$ is a pure quasi-free Hadamard state on the Kruskal spacetime, invariant under the Killing flow $\Lambda$ and the reflection symmetry 
$\iota$, for the latter see comments after \cite[Proposition 5.4]{San15}. 
More precisely,  $\varphi_o$ is characterized by the fact that 
its restriction to the algebra of  the right wedge $\cW_+$ is a KMS state $\varphi_{\beta_H}$  with respect to $\Lambda$, where $\beta_H$ is the Hawking inverse temperature. 
Furthermore, the restriction of $\varphi_o$  to the algebra of the double wedge $W_-\cup W_+$ coincides with the double-KMS state $\tilde{\varphi}_{\beta_H}$ that is
associated with $\varphi_{\beta_H}$ accordingly to \cite{Kay1985}. 

Moreover,  Kay showed that the double KMS-state $\tilde{\varphi}_{\beta_H}$ satisfies the Bisognano-Wichmann property on $W_-\cup W_+$  \cite{Kay1985}; 
actually this property is transferred to the HHI state $\varphi_o$ because, by \cite[Proposition 5.4]{San15}, the mapping 
\[
W_o(f)\xi_o\mapsto \tilde{W}_{\beta_H}(f)\tilde{\xi}_{\beta_H} \, , \qquad f\in C^\infty_0(\cW_-\cup\cW_+) \ ,
\]
gives a unitary operator between the Fock space of $\varphi_o$ and that of the double KMS state $\tilde{\varphi}_{\beta_H}$.
By this unitary equivalence and the fact that $\tilde{\varphi}_{\beta_H}$ satisfies the Reeh-Schlieder property  \cite[Theorem 3.5]{San15}, it follows that 
$\varphi_o$ satisfies the Reeh-Schlieder property too.  
\end{proof}

\smallskip

\noindent
{\bf Acknowledgements.} 
We are grateful to M. S\'anchez for fruitful discussions and an explanation about the proof of Lemma 
\ref{equiv-wedge}. We are also grateful to K. Sanders for fruitful comments and for 
having pointed out an imprecision in Section \ref{relspacetime:C}.

\smallskip\noindent
We acknowledge the MIUR Excellence Department Project awarded to the Department of Mathematics, University of Rome Tor Vergata, CUP E83C18000100006.

\end{document}